\documentclass[a4paper]{article}
\usepackage[utf8]{inputenc}
\usepackage{amsmath}
\usepackage{amsfonts}
\usepackage{amssymb}
\usepackage{amsthm}
\usepackage{hyperref}
\usepackage{graphicx}
\usepackage{mathtools}
\usepackage{xspace}
\usepackage{thmtools,thm-restate}
\usepackage[top=3cm]{geometry}

\newcommand{\ED}{\mathsf{ED}}
\newcommand{\EQ}{\mathsf{EQ}}
\newcommand{\OR}{\mathsf{OR}}

\newcommand{\DISJ}{\mathsf{DISJ}}
\newcommand{\Tribes}{\mathsf{Tribes}}
\newcommand{\Weights}{\mathsf{Weights}}
\newcommand{\cost}{\text{cost}}

\DeclareMathOperator\adeg{\widetilde{\deg}}
\def\Z{\mathbb{Z}}
\def\bZ{\mathbb{Z}}
\DeclarePairedDelimiter\ceil{\lceil}{\rceil}
\DeclarePairedDelimiter\floor{\lfloor}{\rfloor}

\def\Erdos{Erd\H{o}s\xspace}

\newtheorem*{claim*}{Claim}
\newtheorem*{thm*}{Theorem}

\theoremstyle{definition}
\newtheorem{theorem}{Theorem}[section]
\newtheorem{prob}{Problem}[section]
\newtheorem{lemma}[theorem]{Lemma}

\newtheorem{prop}[theorem]{Proposition}

\newtheorem{fact}[theorem]{Fact}

\newtheorem{cor}[theorem]{Corollary}
\newtheorem{claim}{Claim}

\author{Farzan Byramji\thanks{IIT Kanpur. \texttt{farzan@iitk.ac.in}}}
\title{Query complexity of Boolean functions on slices}
\date{\vspace{-20pt}}
\begin{document}
\maketitle
\begin{abstract}
We study the deterministic query complexity of Boolean functions on slices of the hypercube. The $k^{th}$ slice $\binom{[n]}{k}$ of the hypercube $\{0,1\}^n$ is the set of all $n$-bit strings with Hamming weight $k$. We show that there exists a function on the balanced slice $\binom{[n]}{n/2}$ requiring $n - O(\log \log n)$ queries. We give an explicit function on the balanced slice requiring $n - O(\log n)$ queries based on independent sets in Johnson graphs. On the weight-2 slice, we show that hard functions are closely related to Ramsey graphs. Further we describe a simple way of transforming functions on the hypercube to functions on the balanced slice while preserving several measures.
\end{abstract}

\section{Introduction}
Boolean functions, with domain $\{0,1\}^n$ and range $\{0,1\}$, are well-studied objects in theoretical computer science and combinatorics. Recently there has been much interest in understanding Boolean-valued functions on other domains such as the slice (see, e.g., \cite{o2013kkl}), the multislice (see \cite{braverman2022invariance}), the symmetric group (see \cite{filmus2020hypercontractivity}), the Grassmann scheme (see \cite{khot2018pseudorandom}) and high-dimensional expanders (see \cite{gur2022hypercontractivity}). Dafni et al. \cite{dafni2021complexity} gave a general framework which captures several domains of interest and proved polynomial relations between complexity measures for functions on these domains. 

In this paper, we are interested in Boolean-valued functions on the slice. The $k^{th}$ slice of the hypercube $\{0,1\}^n$ is the set of all $n$-bit strings with Hamming weight $k$, where the Hamming weight of a Boolean string $x$ is the number of ones in it. This will be denoted by $\binom{[n]}{k}$ and sometimes be called the slice of weight $k$. We shall assume $k \leq n/2$ throughout, since otherwise we may consider the slice $\binom{[n]}{n-k}$ which has the same properties as $\binom{[n]}{k}$. 

Much work has been done on extending results from Fourier analysis on the hypercube to the slice and these have applications to combinatorics and theoretical computer science. These include the Goldreich-Levin theorem \cite{rao2021fourier}, the Kindler-Safra theorem \cite{filmus2018invariance, keller2020structure}, an invariance principle \cite{filmus2018invariance, filmus2019harmonicity}, the Friedgut-Kalai-Naor theorem \cite{filmus2016friedgut} and the Nisan-Szegedy theorem \cite{filmus2019boolean}. In combinatorics, these have been applied to show a robust Kruskal-Katona theorem \cite{o2013kkl} and a stability version of the \Erdos -Ko-Rado theorem \cite{filmus2018invariance}.

Here we focus on the query complexity of Boolean functions on slices. The deterministic query complexity of a function $f$, $D(f)$,  is the minimum number of adaptive queries that a deterministic query algorithm, which computes $f(x)$ correctly for every $x$ in the domain, must make on the worst-case input (see Section \ref{sec:prel} for the formal definition).
We investigate the \textit{maximum query complexity} of Boolean functions on slices, which is the quantity $\max_f D(f)$ where the maximum is over all functions $f: \binom{[n]}{k} \rightarrow \{0,1\}$.

On the Boolean hypercube, the question of maximum query complexity is well-understood since there are several natural functions (like Parity and other symmetric functions) which require $n$ queries (which is a trivial upper bound). On any slice however, since the weight of the input is fixed, it is sufficient to make just $n-1$ queries to compute any function. A little more thought shows that for $n \geq 3$, any slice function can be computed using $n-2$ queries (see Proposition \ref{prop:upbound}). Still, it is easy to give examples of functions requiring $\Omega(n)$ queries on any slice. Thus we are interested in the following question: how close can the maximum deterministic query complexity on a slice be to $n$?

We show lower bounds for the maximum query complexity of Boolean functions on the balanced slice $\binom{[n]}{n/2}$.
The usual counting argument of comparing the number of functions and the number of decision trees of some depth $d$ only gives a lower bound of $n - O(\log n)$. Moreover any lower bound which is only based on the number of certificates required to cover all inputs cannot beat this bound since the number of inputs is just $\binom{n}{n/2} = O(\frac{2^n}{\sqrt{n}})$. We show the following by a probabilistic argument.

\begin{restatable}{theorem}{balbound} \label{thm:balbound}
There exists a function $f: \binom{[n]}{n/2} \rightarrow \{0,1\}$ with $D(f) \geq n - O(\log \log n)$.
\end{restatable}

\noindent This is based on the idea that a low-depth decision tree for a function $f$ implies the existence of a small balanced certificate, where a certificate is said to be balanced if it contains an equal number of zeroes and ones. Then we use the Lovasz local lemma to show the existence of functions which do not have any such small balanced certificates. 

While we do not have an explicit function achieving this bound, we observe that known independent sets in Johnson graphs can be used to show the following.

\begin{restatable}{theorem}{expbalbound} \label{thm:expbalbound}
There is an explicit function $f: \binom{[n]}{n/2} \rightarrow \{0,1\}$ with $D(f) \geq n - O(\log n)$.
\end{restatable}

We also consider constant weight slices, in particular weight-$1$ and weight-$2$ slices. On the weight-$2$ slice $\binom{[n]}{2}$, we show that hard functions are equivalent to Ramsey graphs, which implies the following. 
\begin{restatable}{prop}{wttwobound}
The maximum deterministic query complexity over functions on the weight-2 slice is $n - \Theta(\log n)$.
\end{restatable}

Since we do not have an explicit function on the balanced slice matching the lower bound in Theorem \ref{thm:balbound}, it is natural to look for a class of functions where such a hard function may be found. On the hypercube, symmetric functions form an interesting class of functions with full query complexity. Since symmetric functions become trivial on a slice, a natural class of functions to consider next is the class of functions obtained by composing two symmetric functions. We show that on any slice for any such function, we can save many queries.

\begin{restatable}{theorem}{compsym} \label{thm:compsym}
Let $f: \{0,1\}^n \rightarrow \{0,1\}$, $g: \{0,1\}^m \rightarrow \{0,1\}$ be symmetric boolean functions. Let $N = nm$. Then on any slice of $\{0,1\}^{nm}$,
\[
D(f \circ g) \leq \min\left\{(n-1)m, nm - \frac{n}{m} \right\} \leq N - N^{1/3}.
\]
\end{restatable}

On the hypercube, there is an important line of work on understanding polynomial relations and separations between various complexity measures (see, for instance, Table $1$ in \cite{aaronson2021degree}). On the balanced slice, polynomial relations between complexity measures were proved by Dafni et al. \cite{dafni2021complexity}. With the aim of giving separations between measures on the balanced slice, we show that a simple way of transforming functions from the hypercube to functions on the balanced slice preserves several measures.

\begin{theorem} \label{thm:lift}
For every function $g : \{0,1\}^n \rightarrow \{0,1\}$, there exists a function $f: \binom{[2n]}{n} \rightarrow \{0,1\}$ such that
\begin{enumerate}
\itemsep -0.25em
\item $D(f) = D(g)$, $R(f) = R(g)$, $R_0(f) = R_0(g)$, $Q(f) = Q(g)$, $Q_E(f) = Q_E(g)$,
\item $\deg(f) = \deg(g)$, $\adeg(f) = \adeg(g)$,
\item $C(f) = C(g)$, $bs(f) = bs(g)$, $s(g) \leq s(f) \leq bs_2(g) \leq 2s(g)^2$.
\end{enumerate}
\end{theorem}
\noindent The above theorem can be used with known separations on the hypercube to get corresponding separations between the above complexity measures on the balanced slice.

This paper is organized as follows. In Section \ref{sec:prel}, we fix notation and define the relevant complexity measures. In Section \ref{sec:con}, we look at the maximum query complexity on constant weight slices, specifically weights $1$ and $2$. In Section \ref{sec:bal}, we prove lower bounds for the maximum query complexity on the balanced slice and also consider the query complexity of Equality and Element Distinctness when restricted to the balanced slice. In Section \ref{sec:comp}, we prove Theorem \ref{thm:compsym}. In Section \ref{sec:lift}, we prove Theorem \ref{thm:lift}. In Section \ref{sec:remarks}, we conclude with some unanswered questions and discuss connections with Ramsey theory.

\section{Preliminaries}\label{sec:prel}
In this section, we define the measures that are studied here. For functions on slices, the definitions below are consistent with those in \cite{dafni2021complexity}.

The set $\{1, 2, \dots, n\}$ is denoted by $[n]$. For a given $n$ and $1 \leq k \leq n-1$, let $\binom{[n]}{k} := \{x \in \{0,1\}^n \mid |x| = \sum_{i = 1}^n x_i = k\}$. We will often interpret a binary string $x \in [n]$ as the subset of $[n]$ containing the positions of $1$s in $x$. More generally for a finite set $S$ and a nonnegative integer $k$, $\binom{S}{k} = \{T \subseteq S \mid |T| = k\}$. The Johnson graph $J_{n, k}$ has $\binom{[n]}{k}$ as its vertex set and two strings are adjacent if their Hamming distance is exactly $2$.

We shall also need the notion of composing two Boolean functions. If we have $f: \{0,1\}^m \rightarrow \{0,1\}$ and $g : \{0,1\}^n \rightarrow \{0,1\}$, then $f \circ g: \{0,1\}^{mn} \rightarrow \{0,1\}$ takes $m$ disjoint inputs $z_1$, $z_2$, $\dots$, $z_m$ each of length $n$ and $(f \circ g) (z_1, z_2, \dots, z_m) = f(g(z_1), g(z_2), \dots, g(z_m))$.

\vspace{10pt}

\noindent \textbf{Decision tree or query complexity.} A \emph{deterministic decision tree} $A$ on $n$ variables is a binary tree in which each leaf is labeled by either a $0$ or a $1$, and each internal node is labeled with an index $i$ for some $i \in [n]$. For every internal node of $A$, one of the two outgoing edges is labeled $0$ and the other edge is labeled $1$. On an input $x$, for an internal node labeled with $i$, it follows the edge labeled by $x_i$ and the output is the label of the leaf reached. Say that the decision tree $A$ computes $f$ if $A(x) = f(x)$ for all $x$ in the domain of $f$.

The cost of $A$ on $x$, $\cost(A, x)$, is the number of queries made by $A$ on $x$. The cost of $A$, $\cost(A)$, is $\max_x \cost(A, x)$, the maximum cost of $A$ on any input. The \emph{deterministic query complexity} of $f$, $D(f)$, is $\min \cost(A)$, where the minimum is over all $A$'s computing $f$.

\vspace{5pt}\noindent \textbf{Certificate complexity.} Let $C$ be a string in $\{0,1,*\}^n$. We shall call $C$ an \emph{assignment}. The \emph{size} of an assignment is the number of non-$*$'s in it. Say that $C$ is \emph{consistent} with $x \in \{0,1\}^n$ if for all $i \in [n]$ where $C_i \neq *$, $C_i = x_i$. For $b \in \{0,1\}$, $C$ is a \emph{$b$-certificate} of a function $f$ if for all strings $x$ in the domain of $f$ which are consistent with $C$, $f(x) = b$. For a string $x$ in the domain of $f$, $C(f, x)$ is the smallest size of an $f(x)$-certificate consistent with $x$. The \emph{certificate complexity} of $f$, $C(f)$, is defined to be the maximum $C(f, x)$ over all strings $x$ in the domain of $f$.

\vspace{5pt}\noindent \textbf{Unambiguous certificate complexity.} Let $\mathcal{C}$ be a collection of certificates of $f$. $\mathcal{C}$ is said to be an \emph{unambiguous} set of certificates if for each $x$ in the domain of $f$, there is a unique $C \in \mathcal{C}$ which is a certificate of $x$. The complexity of an unambiguous set $\mathcal{C}$ of certificates is the largest size of a certificate in $\mathcal{C}$. The \emph{unambiguous certificate complexity} of $f$, $UC(f)$, is the least complexity of an unambiguous collection of certificates for $f$. Note that when considering functions on a slice (or any partial function on the hypercube), there may be a string in the hypercube outside the domain which can be consistent with multiple certificates in the unambiguous collection.

\vspace{5pt}\noindent \textbf{Subcube partition complexity.} A subcube partition of $f$ is a collection of pairs $(C, b)$ where $C$ is a $b$-certificate of $f$ and the collection of the subcubes corresponding to these certificates $C$ forms a partition of the hypercube. Note that we allow certificates which are not consistent with any inputs in the domain. In such a case, the associated $b$ can be arbitrary. The complexity of a subcube partition is the largest size of a certificate in the partition. The \emph{subcube partition complexity} of $f$, $SC(f)$, is the least complexity of a subcube partition of $f$.

Subcube partition complexity is only defined for functions whose domain is a subset of the hypercube, unlike the other measures described here which can be defined for other domains as done in \cite{dafni2021complexity}.

It is easy to see that a subcube partition also gives an unambiguous collection of certificates with the same complexity. The converse need not be true in general because of what was mentioned just before the definition of subcube partition complexity. For (total) functions $f$ on the hypercube, $UC(f)$ and $SC(f)$ coincide, which is why in the literature they are often used interchangeably. Here we keep them separate since $UC(f)$ can be smaller than $SC(f)$ for a slice function $f$. As an example, $UC(f)$ for any non-constant function on a weight-$1$ slice is $1$. (More generally for any function $f$ on $\binom{[n]}{k}$ with $k \leq n/2$, $UC(f) \leq k$.) On the other hand, for the function $f: \binom{[n]}{1} \rightarrow \{0,1\}$ defined by $f(x_1, x_2, \dots, x_n) = \OR(x_1, x_2, \dots, x_{n/2})$ (where $n$ is even for simplicity), $SC(f) = n/2$. This can be seen by considering the subcube containing $0^n$. Each of the indices corresponding to a $1$-input must be part of this assignment. 

On the balanced slice, we have a relation $SC(f) \leq 4UC(f)^2$ using the results of Dafni et al. \cite{dafni2021complexity}. We do not know if this relation is tight, but we have an example ($\ED$ considered in Section \ref{subs:exp}) showing that the two measures can at least differ by roughly a factor of $2$.

\vspace{5pt}\noindent \textbf{Sensitivity and block sensitivity.} 
For a function $f$, a \emph{sensitive block} $B$ of a string $x$ in the domain of $f$ is a subset of $[n]$ such that $x^B$ lies in the domain of $f$ and $f(x) \neq f(x^B)$ where $x^B$ denotes the string obtained by flipping all the bits at positions in $B$. An $l$-block is a block of size at most $l$. The block sensitivity of $f$ on an input $x$, $bs(f, x)$ is the maximum $b$ such that there are $b$ disjoint sensitive blocks of $x$.
The \emph{block sensitivity} of $f$, $bs(f)$, is the maximum $bs(f, x)$ over all $x$ in the domain of $f$. Define $l$-block sensitivity $bs_l(f)$ analogously where we only allow sensitive blocks of size at most $l$.

For a function $f$ on the hypercube, the sensitivity $f$ on an input $x$, $s(f, x)$ is the maximum $b$ such that there are $b$ disjoint sensitive $1$-blocks of $x$. The \emph{sensitivity} of $f$, $s(f)$, is the maximum $s(f, x)$ over all $x$.

For a function $f$ on a slice, the sensitivity $f$ on an input $x$, $s(f, x)$ is the maximum number $b$ such that there are $b$ disjoint sensitive $2$-blocks of $x$. Note that such a block must consist of a $0$ and a $1$. The \emph{sensitivity} of $f$, $s(f)$, is the maximum $s(f, x)$ over all $x$ in that slice.

\vspace{5pt}\noindent \textbf{Degree and approximate degree.} A polynomial $p(x)$ (with variables $x_1, x_2, \dots, x_n$) represents a function $f$ if for all $x$ in the domain of $f$, $f(x) = p(x)$. The \emph{degree} of a function $f$, $\deg(f)$, is the smallest degree of a polynomial representing $f$. A polynomial $p(x)$ approximately represents $f$, if for all $x$ in the domain of $f$, $|f(x) - p(x)| \leq \frac{1}{3}$. The \emph{approximate degree} of a function $f$, $\adeg(f)$, is the smallest degree of a polynomial approximately representing $f$. For functions on the hypercube, the unique multilinear polynomial representing $f$ has the smallest degree. For functions on a slice, multilinearity is not sufficient to guarantee uniqueness. However any function on a slice has a canonical representation as a \emph{harmonic} multilinear polynomial (see \cite{filmus2016orthogonal}), which also has the lowest degree among all representing polynomials.

\vspace{5pt}\noindent For other measures not defined here, refer to the survey by Buhrman and de Wolf \cite{buhrman2002complexity}.

\section{Constant weight slices}\label{sec:con}
In this section, we look at functions on the weight-$1$ slice and the weight-$2$ slice.

\subsection{Weight-1 slice}
On the weight-$1$ slice, we can identify each string with the index of the unique $1$ in it. The following fact is easy, but useful for later proofs.

\begin{prop}\label{prop:wt1dt}
For any $f: \binom{[n]}{1} \rightarrow \{0,1\}$, $D(f) \leq \floor{\frac{n}{2}}$.
\end{prop}
\begin{proof}
It suffices to query the indices in the smaller of $f^{-1}(0)$ and $f^{-1}(1)$.
\end{proof}

This is tight by considering $f$ defined by $f(x, y) = \OR(x)$  where $x$ has $\floor{\frac{n}{2}}$ bits.

Using Proposition \ref{prop:wt1dt}, we next see that one can always save at least two bits to compute a function on a slice.

\begin{prop} \label{prop:upbound}
For all $n \geq 3$, all $0 \leq k \leq n$, $f: \binom{[n]}{k} \rightarrow \{0,1\}$, $D(f) \leq n - 2$.
\end{prop}
\begin{proof}
By Proposition \ref{prop:wt1dt}, any slice function on $3$ bits can be computed using a single query.

Now consider any $f: \binom{[n]}{k} \rightarrow \{0,1\}$. Query the first $n-3$ bits of the input. Then use the optimal decision tree for the restriction of the function to the already seen bits. This uses at most $n-2$ queries.
\end{proof}

\subsection{Weight-2 slice}
Every function $f: \binom{[n]}{2} \rightarrow \{0,1\}$ can be viewed as a graph $G_f$ with vertices $V = [n]$ and $\{i, j\}$ is an edge if $f(\{i, j\}) = 1$. 

For a graph $G = (V, E)$, let $m(G)$ denote the size of the largest monochromatic set, where a set $S \subseteq V$ is said to be monochromatic if it forms a clique or an independent set in $G$.

\begin{lemma}\label{prop:wt2dt}
For every $f: \binom{[n]}{2} \rightarrow \{0,1\}$, 
\begin{align*}
D(f) \geq n - m(G_f).
\end{align*}
\end{lemma}
\begin{proof}
Consider an optimal decision tree computing $f$. Let $S$ be the set of variables which are \emph{not} queried by this tree on the path from the root always taking the edge labelled $0$. Note that $|S| \geq n - D(f)$ since at most $D(f)$ queries could have been made on this path. Since the function is constant for all edges lying entirely in $S$, $S$ is either an independent set or a clique in $G_f$. So $m(G_f) \geq |S|$. Combining this with $|S| \geq n - D(f)$ proves the statement.
\end{proof}

This shows that to find a hard function, it is sufficient to find a graph $G$ with small $m(G)$ which is the well-studied problem of constructing Ramsey graphs. By the famous probabilistic argument of \Erdos , we know that $m(G)$ can be as small as $O(\log n)$.

\begin{theorem}[\cite{erdos1947some}]
There exists a graph $G$ on $n$ vertices with $m(G) < 2 \log n$. In fact, most graphs have this property.
\end{theorem}

The best explicit construction of Ramsey graphs was given by  Chattopadhyay and Zuckerman \cite{chattopadhyay2016explicit} and independently by Cohen \cite{cohen2016two}.
\begin{theorem}[\cite{chattopadhyay2016explicit,cohen2016two}]
There exists an explicit graph $G$ on $n$ vertices with $m(G) < 2^{(\log \log n)^{O(1)}}$.
\end{theorem}
\noindent This gives an explicit function $f$ on the weight-$2$ slice with $D(f) \geq n - 2^{(\log \log n)^{O(1)}}$.

We also have an upper bound on $D(f)$ in terms of $m(G_f)$, based on a decision tree which checks whether the input lies in this large clique or independent set.
\begin{lemma}
For any $f: \binom{[n]}{2} \rightarrow \{0,1\}$,
\begin{align*}
D(f) \leq n - \frac{m(G_f)}{2}.
\end{align*}
\end{lemma}
\begin{proof}
Let $T$ be a largest clique or independent set in $G_f$, so that $|T|= m(G_f)$. Without loss of generality, we may assume that $T$ is a clique. Let $S := [n]\setminus T$.

We now give an algorithm which makes at most $n - \frac{m(G_f)}{2}$ queries. While there is an unqueried variable in $S$ and no $1$ has been set yet, query a variable from $S$. If a $1$ is seen, then use an optimal decision tree for the current restriction of the function (which is on a weight-$1$ slice). Otherwise if all variables in $S$ have been queried without seeing a $1$, output $1$ since $T$ is a clique. 

Now we bound the number of queries made. If the first $1$ is seen on the $d^{th}$ query (where $1 \leq d \leq |S|$), the maximum number of queries made in such a situation is $d + \frac{n-d}{2}$, where the second term accounts for an optimal decision tree on a weight-$1$ slice with $n-d$ bits (Proposition \ref{prop:wt1dt}). This quantity is maximized when $d = |S|$, and it gives $\frac{n + |S|}{2}$. If no $1$ is seen, then $|S|$ queries have been made. So the maximum number of queries made by the algorithm is at most $\frac{n + |S|}{2}$. Since $|S| = n - m(G_f)$, we have $D(f) \leq n - \frac{m(G_f)}{2}$.
\end{proof}

The two lemmas together show that a function on the weight-$2$ slice with high query complexity is essentially a Ramsey graph. In \cite{erdos1935combinatorial}, the bound $m(G) \geq \frac{\log n}{2}$ for every graph $G$ on $n$ vertices is shown for Ramsey's theorem \cite{ramsey}. Using this, we conclude the following.

\wttwobound*

\section{The Balanced Slice}\label{sec:bal}
For even $n$, the balanced slice of the $n$-cube $\{0,1\}^n$ is $\binom{[n]}{n/2}$. In this section, we investigate the query complexity of functions on the balanced slice.

We start by considering the well-studied Equality function in Subsection \ref{subs:eq} to give an example of how non-trivial savings are possible on the balanced slice. We then proceed to give lower bounds on the maximum query complexity on the balanced slice.
By a probabilistic argument, we prove the existence of a function $f$ with $D(f) \geq n - O(\log \log n)$ in Subsection \ref{subs:lll}. In Subsection \ref{subs:exp} we give an explicit function $f$ with $D(f) \geq n - O(\log n)$ based on independent sets in Johnson graphs.

\subsection{Equality}\label{subs:eq}
The equality function $\EQ$ is defined as follows: $\EQ(x, y)$ (where both $x$ and $y$ have length $n/2$) is $1$ if $x = y$. Due to the weight constraint, this function is non-trivial on the balanced slice only when $n \equiv 0 \mod{4}$. So assume $n = 4k$ from now on.

\begin{prop}
$D(\EQ_{4k}) = 3k-1$.
\end{prop}
\begin{claim}
$D(\EQ_{4k}) \leq 3k-1$.
\end{claim}
\begin{proof}
The query algorithm does the following:
\begin{enumerate}
\item Query all $x_i$ for $i \in [2k-1]$. Let $b \in \{0, 1\}$ be the more frequent bit among these. If this bit does not have frequency exactly $k$, reject. 
\item Let $I = \{i \in [2k-1] \mid x_i = b\}$. Query all $y_j$ for $j \in I$. Accept if each of these is $b$. Otherwise reject.
\end{enumerate}
The algorithm makes $2k-1$ queries in step 1. If it doesn't reject in the first step, $|I| = k$ and hence $k$ queries are made in step 2. Thus it makes at most $3k-1$ queries.

Now we show correctness. Suppose the algorithm accepts the input $(x, y)$. Then we have $x_i = y_i = b$ for all $i \in I$. Moreover since $|I| = k$ as it did not reject in step 1, these are all the bits in the input with value $b$. So all the other bits must be $1-b$ and we have $x_j = y_j = 1-b$ for all $j \in [2k] \setminus I$. Thus the algorithm accepts only inputs $(x, y)$ with $x=y$.

Conversely consider how the algorithm proceeds on an input $(x, y)$ with $x = y$. Since $|x| + |y| = 2|x| = 2k$, $x$ contains $k$ 0s and $k$ 1s. So the first $2k-1$ bits of $x$ contain $k$ bits with one value and $k-1$ with the other. Hence the algorithm goes to the next step. Now it accepts since $x_j = y_j = b$ for all $j \in I$. In this manner, $\EQ$ can be computed using $3k-1$ queries.
\end{proof}

\begin{claim}
$D(\EQ_{4k}) \geq 3k-1$.
\end{claim}
\begin{proof}
We give an adversary argument showing that $3k-2$ queries do not suffice to determine the function value. The adversary maintains a string $z \in \{0,1,*\}^{2k}$ which is initially $*^{2k}$ and a bit $b$ that is initially $0$. Here $*$ is meant to signify that the corresponding bits in both $x$ and $y$ haven't been queried yet. When $x_i$ or $y_i$ are queried, the adversary does the following:
\begin{itemize}
\itemsep-0.5em
\item If $z_i = *$, set $z_i := b$ and then $b := 1-b$.
\item Answer $z_i$.
\end{itemize}

In this way, the adversary ensures that the revealed bits of $x$ and $y$ are consistent with $z$. Moreover since $b$ is flipped each time a bit of $z$ is set, around half of the set bits of $z$ are $0$. After the query algorithm finishes, the remaining unset bits of $z$ (if any) can, therefore, be set suitably to get a string $z'$ having weight $k$. Fix any such $z'$. (This $z'$ will also be used for finding a consistent 0-input below.) This implies that there is a 1-input consistent with the revealed bits of $x$ and $y$, namely $(z', z')$.

To see that there is a 0-input consistent with the revealed bits, we consider cases on the number of $*$'s in $z$ (which we shall denote by $l$) at the end of the query algorithm:
\begin{itemize}
\item $l \geq 2$

Since the bits of $z$ are set to $0$ and $1$ in an alternating fashion and at least $2$ are not yet set, there exist $i, j \in [2k]$ such that $z_i = z_j = *$, $z'_i = 0$ and $z'_j = 1$. In this case, $(z', z'^{(i \; j)})$ is a 0-input. It is consistent with the revealed bits of $x$ and $y$ since none of $y_i, y_j$ have been queried.

\item $l \leq 1$

In this case, since at most $3k-2$ queries are made and at least $2k-1$ bits of $z$ are set, the number of indices $i \in [2k]$ such that both $x_i$ and $y_i$ have been queried can be at most $(3k-2) - (2k-1) = k-1$. Since there are $k$ 0's and $k$ 1's in $z'$, there must exist $i, j \in [2k]$ such that $z'_i = 0, z'_j = 1$, it is not the case that both $x_i, y_i$ have been queried and similarly it is not the case that both $x_j, y_j$ have been queried. Then we can get a consistent $0$-input by swapping the unqueried bit among $x_i, y_i$ with the unqueried bit among $x_j, y_j$. In more detail, we have the following.
\begin{itemize}
\item If $x_i$ and $x_j$ have not been queried, then $(z'^{(i \; j)}, z')$ is a consistent 0-input. Similarly if $y_i$ and $y_j$ have not been queried, then $(z', z'^{(i \; j)})$ is a consistent 0-input.
\item If $x_i$ and $y_j$ have not been queried, then $(z'^{(i)}, z'^{(j)})$ is a consistent 0-input and similarly when $x_j$ and $y_i$ have not been queried. \qedhere
\end{itemize}
\end{itemize}
\end{proof}

\subsection{An $n - O(\log \log n)$ lower bound}\label{subs:lll}
Filmus \cite{filmus2022com} introduced the notion of balanced certificates. Let $f: \binom{[n]}{n/2} \rightarrow \{0,1\}$ be a Boolean function on the balanced slice. Let $C \in \{0, 1, *\}^n$ be a certificate for $f$. $C$ is said to be balanced if $|\{i \mid C_i = 0\}| = |\{i \mid C_i = 1\}|$. For $x \in \binom{[n]}{n/2}$, define the balanced certificate complexity of $f$ at $x$, $BC(f, x) = \min_C |C|$ where $C$ varies over all balanced certificates of $f$ consistent with $x$. The balanced certificate complexity of $f$ is defined by $BC(f) = \max_{x \in \binom{[n]}{n/2}} BC(f, x)$. The minimum balanced certificate complexity of $f$ is $mBC(f) = \min_{x \in \binom{[n]}{n/2}} BC(f, x)$. 

$BC(f)$ can be as large as $n$. This happens precisely when there is some input which is sensitive on each of the $(n/2)^2$ transpositions exchanging a $0$ and a $1$. An example of such a function is $\EQ$. On the other hand, the minimum balanced certificate complexity of $\EQ$ is just $2$.

\begin{prop}\label{prop:mbc}
For any $f: \binom{[n]}{n/2} \rightarrow \{0,1\}$, $D(f) \geq mBC(f) - 1$.
\end{prop}
\begin{proof}
Let $T$ be a decision tree for $f$. Consider the path obtained by starting at the root of $T$, then alternately picking the branch with $0$ followed by the branch with $1$ repeatedly until we reach a leaf. This defines a certificate, say $C$, for $f$. Let $c_0$ denote $|\{i \mid C_i = 0\}|$ and similarly define $c_1$. Then $c_1 \leq c_0 \leq c_1 + 1$. So $C$ is nearly balanced. If $c_0 = c_1 + 1$, we can pick a bit $i$ such that $C(i) = *$ arbitrarily and get a balanced certificate $C'$ defined by $C'(i) = 1$ and $C'(j) = C(j)$ for all $j \neq i$. This shows that $f$ has a balanced certificate of size at most one more than the depth of $T$. Taking an optimal decision tree $T$ finishes the proof.
\end{proof}

We note that all explicit functions $f$ considered in this section have $mBC(f) \leq \frac{n}{2} + O(1)$, and so the above proposition does not give anything interesting for them. However we can show by the probabilistic method that functions with large $mBC$ exist.

\begin{lemma}\label{lemma:lll}
There exists a function $f: \binom{[n]}{n/2} \rightarrow \{0,1\}$ with $mBC(f) \geq n - O(\log \log n)$.
\end{lemma}
\begin{proof}
We will use the Lov\'{a}sz Local Lemma (see, for instance, \cite[Chapter 5]{alon2016probabilistic}). Let $h$ be a nonnegative integer to be fixed later. For a suitable choice of $h$, we will show using the local lemma that there is a function $f$ with no balanced certificate of size $n - 2h$. For ease of notation, set $k = n/2$.

Pick $f$ uniformly at random. For each balanced restriction $C$ fixing $n - 2h$ bits, let $E_C$ denote the event that $C$ is actually a certificate for $f$. In other words, $E_C$ occurs if and only if all the $\binom{2h}{h}$ inputs consistent with $C$ have the same function value under $f$. Then for every $C$, $p := \Pr[E_C] = \frac{2}{2^{\binom{2h}{h}}}$.

Now we bound the degree of the dependency graph on these events $E_C$. $E_{C_1}$ can depend on $E_{C_2}$ only if there is an input which is consistent with both $C_1$ and $C_2$. So to bound the degree of $E_C$ for any balanced $C$, pick any input $x$ consistent with $C$ and then a balanced certificate for $x$. The former can be done in $\binom{2h}{h}$ ways since $C$ is balanced and the latter can be done in $\binom{k}{h}\cdot\binom{k}{h}$ ways. 
So if $d$ denotes the maximum degree of the dependency graph, we have $d+1 \leq \binom{2h}{h} \binom{k}{h}^2$. By the symmetric local lemma, if we pick $h$ such that 

\begin{align*}
    e \cdot \frac{2}{2^{\binom{2h}{h}}} \cdot \binom{2h}{h} \binom{k}{h}^2 \leq 1,
\end{align*}
then there exists an $f$ with no balanced certificate of size $n - 2h$. Taking $h = O(\log \log n)$ suffices.
\end{proof}

\balbound*
\begin{proof}
This follows by combining Lemma \ref{lemma:lll} and Proposition \ref{prop:mbc}.
\end{proof}

The next proposition gives a Ramsey-like statement for $mBC$.
\begin{prop}
For every $h \geq 1$, there exists $k_0$ such that for all $k \geq k_0$ and all functions $f: \binom{[2k]}{k} \rightarrow \{0,1\}$, $mBC(f) \leq 2(k - h)$.
\end{prop}
\begin{proof}
We will show that $k_0 \leq R_h(2h) - h$, where $R_r(s)$ is the usual Ramsey number for $r$-graphs, i.e. the smallest $n$ such that every coloring of $\binom{[n]}{r}$ with two colours contains a monochromatic complete $r$-graph on $s$ vertices. Let $f: \binom{[2k]}{k} \rightarrow \{0,1\}$ be any function where $k \geq R_h(2h) - h$.

It is convenient to now think of $f$ as a $2$-coloring of $\binom{[2k]}{k}$. Consider the restriction $g$ of $f$ which fixes the last $k-h$ inputs in $[2k]$ to $1$. Then $g$ is a function on the slice $\binom{[k+h]}{h}$ and so can be viewed as a $2$-coloring of $\binom{[k+h]}{h}$. Since $k+h \geq R_h(2h)$, there is a monochromatic complete $h$-graph $H$ on $2h$ vertices in the coloring defined by $g$. This monochromatic $h$-graph $H$ in $\binom{[k+h]}{h}$ can be described by a certificate of $g$ which fixes all $k-h$ vertices outside $H$ to $0$. Combining this with the $k-h$ bits fixed to $1$ gives the required balanced certificate for $f$.
\end{proof}
With this proposition, known upper bounds on Ramsey numbers for hypergraphs (see, e.g., \cite{graham1991ramsey}) imply that $k_0 \leq t_{h-1}(ch)$ where $c$ is an absolute constant and $t_i$ is the tower function defined as $t_0(x) = x$ and for all $i > 0$, $t_i(x) = 2^{t_{i-1}(x)}$. On the other hand, Lemma \ref{lemma:lll} only shows that in the above statement $k_0$ must be at least $2^{2^{\Omega(h)}}$. The upper bound implies that for any function $f$ on the balanced slice, $mBC(f) \leq n - \omega(1)$ where the $\omega(1)$ quantity is smaller than $\log^* n$. 

\subsection{An explicit $n - O(\log n)$ lower bound}\label{subs:exp}

We will use the following standard fact which lower bounds depth by a packing argument.
\begin{fact}\label{prop:cubes}
For a function $f$, suppose $m \geq \max_A |A \cap f^{-1}(1)|$  where $A$ ranges over all $1$-subcubes of $f$. Then $D(f) \geq \log(|f^{-1}(1)|/m)$.
\end{fact}
\begin{proof}
Since each $1$-leaf in a decision tree for $f$ can accept at most $m$ inputs, there must be at least $|f^{-1}(1)|/m$ such leaves. Thus its depth must be at least $\log(|f^{-1}(1)|/m)$.
\end{proof}
\begin{cor}\label{cor:indset}
If $f$ is a Boolean function on the slice $\binom{[n]}{k}$ and $f^{-1}(1)$ is an independent set in the Johnson graph $J(n, k)$, then $D(f) \geq \log(|f^{-1}(1)|)$.
\end{cor}
\begin{proof}
It suffices to show that $\max_A |A \cap f^{-1}(1)|$ is $1$ by the above fact. Towards contradiction, suppose there is a $1$-subcube $A$ and strings $x \neq y$ such that $x, y \in A$. 
As $A$ is a subcube, it contains the subcube generated by $x$ and $y$, i.e. the set of inputs consistent with the assignment containing all the bits where $x$ and $y$ are the same. Let $i$ be an index with $x_{i} = 0, y_{i} = 1$ and $j$ be such that $x_{j} = 1, y_{j}  = 0$. Set $x':= x^{(i j)}$. Then $x'$ lies in $A$ and should be a $1$-input. But this contradicts the fact that $f^{-1}(1)$ is an independent set since $x$ is a neighbour of $x'$.
\end{proof}

Now we use Corollary \ref{cor:indset} to find functions with large query complexity. The independence number of Johnson graphs has been studied under various guises but it is not yet exactly known. However the known bounds are sufficient to describe up to additive constants the best one can do using Corollary \ref{cor:indset}.

The following lower bound on the independence number was shown by Graham and Sloane \cite{graham1980lower}, originally as constant-weight codes.
\begin{theorem}[\cite{graham1980lower}]
There exists an independent set of size at least $\frac{1}{n}\binom{n}{k}$ in $J(n, k)$.
\end{theorem}
We reproduce the short and elegant proof below for completeness.
\begin{proof}
It will be convenient to let the ground set be $\Z_n$. For each $i \in \Z_n$, define 
\[
I_i := \{S \subseteq \Z_n \mid |S| = k, \sum_{x \in S} x = i\}.
\]
The key observation is that each $I_i$ is an independent set. To see this, consider any $S, S' \in I_i$ and suppose $x_1, x_2, \dots, x_{k-1}$ are elements of both $S$ and $S'$. Then the remaining elements $x_k$ of $S$ and $x_k'$ of $S'$ are the same since 
\[
x_1 + x_2 + \dots + x_{k-1} + x_k = i = x_1 + x_2 + \dots + x_{k-1} + x'_k
\]

As the $I_i$'s are disjoint and each $S \in \binom{\Z_n}{k}$ belongs to some $I_i$, $\sum_{i \in \Z_n} |I_i| = \binom{n}{k}$. By averaging,  there exists some $j$ such that $|I_j| \geq \frac{1}{n}\binom{n}{k}$.
\end{proof}

For certain values of $n$ and using a different group, one can get an explicit independent set by exactly computing the cardinalities of these sets, as done by Katona and Makar-Limanov.

\begin{theorem}[\cite{katona2008problem}]
For $n = 2^r$, $k = n/2$ and the group $\bZ_2^r$,
\begin{align*}
|I_0| = \frac{1}{n}\binom{n}{n/2} + \left(1 - \frac{1}{n}\right)\binom{n/2}{n/4}.
\end{align*}
\end{theorem}

This immediately gives the following by Corollary \ref{cor:indset}.

\expbalbound*

We next consider the element distinctness function $\ED$ on the balanced slice. Let $k := 2^l$ where $l$ is a parameter. Then $\ED(x_1, x_2, \dots, x_k) = 1$ if all the $x_i$'s are distinct, where each $x_i$ is on $l$ bits. Note that by the choice of $k$, each $1$-input forms a permutation of $\{0,1\}^l$. Also observe that a string representing such a permutation is balanced. This quickly follows by considering the involution $I$ which maps a string to its complement and seeing that each pair $(x_i, I(x_i))$ is balanced.

\begin{lemma}\label{lemma1}
For $f= \ED$, $\max_A |A \cap f^{-1}(1)| = 2$.
\end{lemma}
\begin{proof}
The lemma easily follows from the following claim (proved later) which describes when two $1$-inputs can be present together in the same $1$-subcube.
\begin{claim*}
If two $1$-inputs of $\ED$ are in the same $1$-subcube, then they must be neighbours. Moreover there exist $s, t \in [k]$ and $j \in [l]$ such that if one input is $(x_1, x_2, \dots, x_k)$, the other input is $(x_1, x_2, \dots, x_s^{(j)}, \dots, x_t^{(j)}, \dots, x_k)$ and $x_s = x_t^{(j)}$.
\end{claim*}
Aiming for a contradiction, suppose $X = (x_1, x_2, \dots, x_k)$, $Y = (y_1, y_2, \dots, y_k)$ and $Z = (z_1, z_2, \dots, z_k)$ lie in a $1$-subcube. Then from the claim it follows that the pairs of coordinates where $X, Y$ disagree and those where $X, Z$ disagree are disjoint. But then $Y$ and $Z$ disagree on all $4$ of these coordinates, contradicting the claim. So we have shown that at most two $1$-inputs can lie in a $1$-subcube.
\end{proof}

\begin{proof}[Proof of claim]
Suppose $X = (x_1, x_2, \dots, x_k)$ and $Y = (y_1, y_2, \dots, y_k)$ are in a $1$-subcube. Let $I := \{(i, j) | x_{i, j} \neq y_{i, j}\}$ be the set of all coordinates where $X$ and $Y$ differ. Fix some $(i_0, j_0) \in I$ with $x_{i_0, j_0} = 0$ (this exists due to the fact that $X$ and $Y$ are balanced). 

Let $(i_1, j_1) \in I$ be any index such that $x_{i_1, j_1} = 1$. We will show that $j_0 = j_1$ and $x_{i_0} = x_{i_1}^{(j_1)}$. First note that $i_0 \neq i_1$. If not, $X' = (x_1, x_2, \dots, x_{i_0}^{(j_0 j_1)}, \dots, x_k)$ should be a $1$-input since $X'$ lies in the subcube generated by $X$ and $Y$. However $X'$ is not a $1$-input since $x_{i_0}^{(j_0\; j_1)} = x_i$ for some $i \neq i_0$ as $X$ is a permutation of $\{0,1\}^l$. So $i_0 \neq i_1$.

Switching the bits at $(i_0, j_0)$ and $(i_1, j_1)$ in $X$ gives another $1$-input, \\
$X'' = (x_1, x_2, \dots, x_{i_0}^{(j_0)}, \dots, x_{i_1}^{(j_1)},\dots, x_k)$. Since this is also a permutation of $\{0,1\}^l$, we must have $x_{i_1} = x_{i_0}^{(j_0)}$ and $x_{i_0} = x_{i_1}^{(j_1)}$. Together this gives $x_{i_0} = x_{i_1}^{(j_1)} = (x_{i_0}^{(j_0)})^{(j_1)}$ from which we deduce that $j_0 = j_1$. Since $(i_1, j_1) \in I$ was arbitrary with $x_{i_1, j_1} = 1$, this shows that there is a unique such $(i_1, j_1)$ given by $j_1 = j_0$ and $i_1$ such that $x_{i_1} = x_{i_0}^{(j_0)}$. By symmetry, $(i_0, j_0)$ is also unique. This finishes the proof of the claim.
\end{proof}

\begin{prop}
$D(\ED) \geq n - O(n/\log n)$
\end{prop}
\begin{proof}
There are $k! = (2^l)!$ $1$-inputs for $\ED$. From Fact \ref{prop:cubes} and Lemma \ref{lemma1}, we have that $D(\ED) \geq \log\left(\frac{(2^l)!}{2}\right)$, which using asymptotics for the Lambert W function gives $n - O(n/\log n)$.
\end{proof}

By using the above ideas, we also see that $\ED$ requires $n-1$ non-adaptive queries for $k \geq 2$. Computer experiments suggest that the above bound on $D(\ED)$ can be improved to around $n - l/2 = n - O(\log n)$. We do have an upper bound of $n-l/2$ using ideas very similar to those in the upper bound for equality. 
More generally, we may consider $\ED_{k, l}: \binom{[kl]}{kl/2} \rightarrow \{0,1\}$ defined as before but where $k$ is not necessarily $2^l$.
\begin{prob}
Is it true that for all $l \geq 2$ and $3 \leq k \leq 2^l$ such that $kl$ is even, $D(\ED_{k,l}) = kl - \ceil{\frac{l}{2}}$?
\end{prob}

\section{Composed functions on slices}\label{sec:comp}
In this section, we look at composed functions $f \circ g$ where both $f$ and $g$ are symmetric boolean functions on $\{0,1\}^n$ and $\{0,1\}^m$ respectively. The following proposition shows that no  function from this class achieves the bounds of the previous sections for hard functions on a slice.

\compsym*
\begin{proof}
We shall show the upper bound for the following problem, which shall imply the above statement. Given $n$ strings $x_1, x_2, \dots x_n$ each on $m$ bits and the total weight $\sum_{i = 1}^n |x_i|$, we wish to find the multiset of individual weights $\{\{ |x_1|, |x_2|, \dots |x_n| \}\}$. An obvious way is to query all $x_j$ for $1 \leq j \leq n-1$. The weight of $x_n$ can now be determined by using the known total weight. This uses $(n-1)m$ queries. We shall call this algorithm $A$. \label{alg:A}

Another way is the following. For each $i \in [n]$, query all bits of $x_i$ except the last. We use $(m-1)n$ queries in doing this. Now we know partial weights $w_i' = \sum_{j = 1}^{m-1} x_{i, j}$ for all $i \in [n]$. Note that $0 \leq w_i' \leq m-1$ for all $i$. Therefore some $w_0' \in \{0, 1, \dots, m-1\}$ is $w_i'$ for at least $n/m$ many $i \in [n]$. Let $S \subseteq [n]$ be the set of such $i$. Now query $x_{i, m}$ for $i \in [n]\setminus S$. This uses at most $n - n/m$ additional queries. This determines the weights of all $x_i$ with $i \in [n]\setminus S$ exactly. For the remaining weights, if the weight of the remaining unqueried bits is $k$, then $k$ strings $x_i, i \in S$ have weight $w_0 + 1$ and the remaining strings in $S$ have weight $w_0$. So we have determined the multiset of individual weights. This query strategy will be called algorithm $B$. \label{alg:B} This totally uses $nm - \frac{n}{m}$ queries.
\end{proof}

Some natural examples of functions of the above kind include $\DISJ, \EQ$ and $\Tribes$. The above statement gives an upper bound of $3N/4$ for $\EQ$ and $\DISJ$. On the balanced slice, this is tight for $\EQ$ up to an additive constant as shown earlier and is exactly tight for $\DISJ$. In Appendix \ref{sec:weights}, we examine whether the above general strategy for the problem $\Weights_{n, m, k}$ of computing the multiset of weights of $x_i$'s is tight. 

\section{Lifting separations from the hypercube to the balanced slice} \label{sec:lift}
In this section, we see how an obvious way to obtain a function on the $2n$-bit balanced slice from a function on the $n$-hypercube preserves many measures. Specifically this allows us to readily obtain separations between measures on the balanced slice using known separations for the hypercube (see Table $1$ in \cite{aaronson2021degree}).

For any $g: \{0,1\}^n \rightarrow \{0,1\}$, we can define a function on the balanced slice $f_g: \binom{[2n]}{n} \rightarrow \{0,1\}$ by $f_g(x, y) = g(x)$ where $x, y$ are each of $n$ bits. We shall drop the subscript $g$ whenever it causes no confusion. Most of the relations are straightforward, but we prove everything for completeness.

\begin{lemma} For any $g: \{0,1\}^n \rightarrow \{0,1\}$ and $f := f_g$
\begin{enumerate}
\item $D(f) \leq D(g)$, $R(f) \leq R(g)$, $R_0(f) \leq R_0(g)$, $Q(f) \leq Q(g)$, $Q_E(f) \leq Q_E(g)$, $\deg(f) \leq \deg(g)$, $\adeg (f) \leq \adeg(g)$
\item $bs(f) \geq bs(g), s(f) \geq s(g)$
\item $C(f) \leq C(g)$, $UC(f) \leq UC(g)$, $SC(f) \leq SC(g)$
\end{enumerate}
\end{lemma}
\begin{proof}
\begin{enumerate}
\item Any query algorithm for $g$ gives a corresponding algorithm for $f$. The same thing holds for polynomials.
\item We will show that for any $(x, y)$, $s(f, (x, y)) \geq s(f, x)$ and $bs(f, (x, y)) \geq bs(f, x)$. First pair up the $0$s in $x$ with the $1$s in $y$ and the $1$s in $x$ with the $0$s in $y$. This can be done since both are of $n$ bits and their weights sum to $n$. So for each $i \in [n]$ if $x_i$ is paired up with $y_{i'}$, we have $x_i \neq y_{i'}$.

Consider disjoint sensitive blocks $b_1, b_2, \dots, b_k$ witnessing $bs(g, x) = k$. For each $i \in [k]$, let $b_i'$ be the block containing the corresponding indices of $y$: $b_i' = \{j' \mid j \in b_i\}$. Then $(x^{b_i}, y^{b_i'})$ is balanced and $f(x, y) \neq f(x^{b_i}, y^{b_i'})$. So for each sensitive block of $x$ for $g$, there is a corresponding block of $(x, y)$ which is sensitive for $f$. So $bs(f) = \max_{(x, y)} bs(f, (x, y)) \geq \max_{(x, y)} bs(g, x) = bs(g)$. By the same argument, we have $s(f) \geq s(g)$.

\item A $g$-certificate $c$ for any $x \in \{0,1\}^n$ gives an $f$-certificate $(c, *^n)$ of $(x, y)$ for every $y \in \{0,1\}^n$. The collection of such certificates for $f$ obtained from unambiguous certificates of $g$ also forms a subcube partition of $\{0,1\}^{2n}$ computing $f$. \qedhere
\end{enumerate}
\end{proof}

The other direction, for most measures, relies on the observation that the balanced slice contains all $(x, \bar{x})$ for $x \in \{0,1\}^n$, where $\bar{x}$ is the bitwise complement of $x$, i.e. $\bar{x}_i = 1-x_i$ for each $i \in [n]$. So for any given $x \in \{0,1\}^n$, one can pretend that the given input is $(x, \bar{x})$ and use the algorithm for $f$ to compute $g(x)$.

\begin{lemma} For any $g: \{0,1\}^n \rightarrow \{0,1\}$ and $f := f_g$, 
$D(f) \geq D(g)$, $R(f) \geq R(g)$, $R_0(f) \geq R_0(g)$, $Q(f) \geq Q(g)$, $Q_E(f) \geq Q_E(g)$, $\deg(f) \geq \deg(g)$, $\adeg(f) \geq \adeg(g)$.
\end{lemma}
\begin{proof}
We show that any decision tree for $f$ also gives a decision tree for $g$. Given a decision tree for $f$, follow it in the natural way if  at a node labelled with some $x_i$. If the node is labelled with $y_i$, query $x_i$, which is say $b$, and follow the path for $y_i = 1-b$. Continue like this until a leaf is reached which gives the output. The arguments for other query complexity measures are essentially the same.

Given a polynomial $p(x, y)$ representing $f$, we obtain a polynomial representing $g$ by considering $q(x) := p(x, \bar{x})$. Put another way, substitute $y_i := 1-x_i$ for all $i \in [n]$. Since an affine substitution cannot increase the degree, we get $\deg(f) \geq \deg(g)$. Similarly $\adeg(f) \geq \adeg(g)$.
\end{proof}

\begin{lemma}
For any $g: \{0,1\}^n \rightarrow \{0,1\}$ and $f := f_g$, $bs(f) \leq bs(g)$.
\end{lemma}
\begin{proof}
Pick some $(x, y)$ such that $bs(f, (x, y)) = bs(f)$. Let $l := bs(f)$. Let $b_1$, $b_2$, $\dots$, $b_l$ be disjoint blocks on which $(x, y)$ is sensitive. Consider the restrictions $b_1', b_2', \dots, b_l'$ of these blocks to only the indices in $x$. We claim that $x$ is sensitive on each of these blocks for the function $g$. 

First note that each $b_i'$ is non-empty, since no block contained entirely in $y$ can be a sensitive block by the definition of $f$. Now it is easy to see that $x$ must be sensitive (for the function $g$) on each $b_i'$ since otherwise $b_i$ would not be a sensitive block of $(x, y)$. So we have $bs(g) \geq bs(g, x) \geq bs(f)$.
\end{proof}

\begin{lemma}
For any $g: \{0,1\}^n \rightarrow \{0,1\}$ and $f := f_g$, $s(f) \leq bs_2(g) \leq 2s(g)^2$.
\end{lemma}
\begin{proof}
The second inequality is the Kenyon-Kutin bound \cite{kenyon2004sensitivity}.

For the first inequality, pick an $(x, y)$ with $l := s(f) = s(f, (x, y))$. Let $b_1$, $b_2$, $\dots$, $b_l$ be disjoint sensitive blocks of $(x, y)$ each of size $2$. Observe that each block either lies entirely in $x$, or contains a bit in $x$ and the other in $y$. As in the proof of the previous proposition, consider the restrictions $b_1'$, $b_2'$, $\dots$, $b_l'$ to the variables in $x$. These have size at most $2$ and $x$ is sensitive (for $g$) on each of these. Therefore $bs_2(g) \geq s(f)$.
\end{proof}

This relation is tight up to multiplicative constants by the following variant of the Rubinstein function \cite{rubinstein1995sensitivity}. Let $h(z)$ be the function on $2k = \sqrt{n}$ bits defined in the following way: $h(z) = 1$ if $z$ is of the form $(01)^i 10 (01)^{k-(i+1)}$ for some $0 \leq i < k$. Define $g = \OR_{\sqrt{n}} \circ h$ where $\circ$ denotes composition. In other words, $g(z_1, z_2, \dots, z_{\sqrt{n}})$ is $1$ only when some $z_j$ is of the form $(01)^i 10 (01)^{k-(i+1)}$. The sensitivity of this function is $\sqrt{n}$. (The analysis is the same as that for the Rubinstein function.) On the other hand, consider $(z^{\sqrt{n}}, y)$ where $z = (01)^{k}$ and $y$ is any balanced string on $n$ bits. This input has sensitivity at least $\frac{n}{2}$ for $f_g$. Flipping any $01$ in any copy of $z$ converts this $0$-input into a $1$-input.

\begin{lemma} For any $g: \{0,1\}^n \rightarrow \{0,1\}$ and $f := f_g$, 
$C(g) \leq C(f)$.
\end{lemma}
\begin{proof}
We show how to convert an $f$-certificate of $(x, y)$ into a $g$-certificate of $x$ without increasing the size of the certificate. Let $C = (C_x, C_y)$ be a certificate of $(x, y)$. Let $c_0$ and $c_1$ be the number of $0$'s and $1$'s respectively in $C_y$. Similarly let $d_0$ and $d_1$ be the number of $0$'s and $1$'s respectively in $C_x$.
If $c_0 > d_1$, arbitrarily pick $c_0 - d_1$ indices $i$ in $x$ where $x_i = 1$ which are not already in $C_x$ and add them to $C_x$. There must exist at least $c_0$ $1$'s in $x$ since otherwise there would be more than $n/2$ $0$'s in $(x, y)$. Similarly if $c_1 > d_0$, arbitrarily pick $c_1 - d_0$ indices in $x$ where $x_i = 0$ and add them. Call this $C_x'$. We claim that $C_x'$ is a $g$-certificate of $x$.

Consider any $x'$ consistent with $C_x'$. We will find $y'$ such that $(x', y')$ is consistent with $C$. Then we will have $g(x') = f(x', y') = f(x, y) = g(x)$. Simply pick $y'$ to be any string with weight $n/2 - |x'|$ which is consistent with $C_y$. To see that such a string exists, we just need to check that $x'$ contains at most $(n/2 - c_0)$ $0$'s and at most $(n/2 - c_1)$ $1$'s. This follows from the fact that $C_x'$ contains at least $c_0$ $1$'s and at least $c_1$ $0$'s. So we have a certificate $C_x'$ of $x$ with $|C_x'| \leq |C|$.

Finally $C(g) = \max_x C(g, x) \leq \max_{(x, y)} C(f, (x, y)) = C(f)$.
\end{proof}

\begin{lemma} For any $g: \{0,1\}^n \rightarrow \{0,1\}$ and $f := f_g$, 
$UC(g) \leq 4UC(f)^2$, $SC(g) \leq 4SC(f)^2$.
\end{lemma}
\begin{proof}
The first inequality follows from combining $UC(g) \leq D(g) \leq D(f) \leq 4C(f)^2$ and $C(f) \leq UC(f)$. $D(f) \leq 4C(f)^2$ follows from the results of Dafni et al. \cite{dafni2021complexity}. The second inequality is similar.
\end{proof}

This bound can quite possibly be improved. We have not been able to find a function $g$ with $SC(f_g) < SC(g)$ or even $UC(f_g) < UC(g)$ so far.

Putting the above lemmas together, we have the following.
\begin{theorem}
For any $g: \{0,1\}^n \rightarrow \{0,1\}$ and $f := f_g$, 
\begin{enumerate} 
\itemsep -0.25em
\item $D(f) = D(g)$, $R(f) = R(g)$, $R_0(f) = R_0(g)$, $Q(f) = Q(g)$, $Q_E(f) = Q_E(g)$,
\item $\deg(f) = \deg(g)$, $\adeg(f) = \adeg(g)$,
\item $bs(f) = bs(g)$, $s(g) \leq s(f) \leq bs_2(g) \leq 2s(g)^2$,
\item $C(f) = C(g)$, $UC(f) \leq UC(g) \leq 4UC(f)^2$, $SC(f) \leq SC(g) \leq 4SC(f)^2$.
\end{enumerate}
\end{theorem}

The above relations imply that any separation for the hypercube not involving $s(f)$ and $UC(f)$ can be directly used to give a separation for the balanced slice. Separations of the form $s(f) \geq M(f)^c$ and $M(f) \geq UC(f)^c$ can still be obtained from those on the hypercube.

Though we cannot use the above relations to prove the separation $bs(f) \geq \Omega(s(f)^2)$ on the balanced slice, we still have it by considering the restriction of the Rubinstein function to the balanced slice.

\begin{prop}
There is a function $f$ on the balanced slice with $bs(f) \geq s(f)^2/16$.
\end{prop}
\begin{proof}
Recall that Rubinstein's function $f$ \cite{rubinstein1995sensitivity} on input $(z_1, z_2, \dots, z_{\sqrt{n}})$ is $1$ if there is an $i \in [\sqrt{n}]$ such that $z_i$ is of the form $(00)^i11(00)^{k-i-1}$ where $k = \sqrt{n}/2$. Then on the balanced slice $bs(f) \geq n/4$ by considering the input with $z_i = 0^{\sqrt{n}}$ for $1 \leq i \leq \sqrt{n}/2$, and $z_j = 1^{\sqrt{n}}$ for $\sqrt{n}/2 + 1 \leq j \leq \sqrt{n}$.

Now we bound $s(f)$ from above. Any $1$-input with a sensitive $2$-block can have at most two $z_i$'s of the form $(00)^i11(00)^{k-i-1}$. Since any such sensitive block with a $0$ and a $1$ must contain an index from every such  $z_i$, there can be at most $\sqrt{n}$ such disjoint sensitive blocks. So the $1$-sensitivity of $f$ can be at most $\sqrt{n}$. 

Now consider a $0$-input $x = (z_1, z_2, \dots, z_{\sqrt{n}})$ and a collection of disjoint sensitive $2$-blocks. We show that any $z_i$ can contribute at most two sensitive $2$-blocks. Here by $z_i$ contributing a sensitive $2$-block we mean  that flipping the block changes $z_i$ to make it of the form $(00)^i11(00)^{k-i-1}$. If $|z_i| \notin \{1,2,3\}$, then flipping a $2$-block cannot change $z_i$ to $(00)^i11(00)^{k-i-1}$. If $|z_i| = 2$, there are exactly two sensitive $2$-blocks in $z_i$. If $|z_i| \in \{1,3\}$, there is at most one sensitive $2$-block which transforms $z_i$ appropriately and this sensitive block intersects $z_i$ in one bit. So the $0$-sensitivity is at most $2\sqrt{n}$.
\end{proof}

\section{Concluding remarks}\label{sec:remarks}
The idea used for the maximum query complexity on the weight-2 slice extends to higher constant weight slices but we no longer have tight bounds. In the known $k$-hypergraph Ramsey theorems, $n$ is upper bounded by roughly a tower of height $k-1$ ending with $m(G)$ (see, e.g., \cite{graham1991ramsey}). Using this, one can show by induction on $k$ that the maximum query complexity on the weight-$k$ slice is at most $n - \Omega(\log^{(\binom{k}{2})} n)$ for each constant $k \geq 2$. In the other direction, \Erdos , Hajnal  and Rado \cite{erdHos1965partition} showed that there is a family of $3$-hypergraphs not containing monochromatic sets of size $\Omega(\sqrt{\log n})$. By the stepping-up lemma of \Erdos and Hajnal (see \cite{graham1991ramsey}), for every $k \geq 3$, there is a family of $k$-hypergraphs not containing any monochromatic sets of size $\Omega(\sqrt{\log^{(k-2)} n})$ (where $\log^{(1)} n = \log n$ and for all $i \geq 2$, $\log^{(i)} n = \log(\log^{(i-1)} n)$).

It would be interesting to improve the bounds for the maximum query complexity on the balanced slice, and find explicit functions requiring $n - o(\log n)$ queries. We now point out the connection between balanced certificates and daisies which were introduced in \cite{bollobas2011daisies}. An $(r, m, k)$-daisy $H$ is a $(k + r)$-uniform hypergraph  on $k + m$ vertices for which there is a partition $V(H) = K \cup M$ with $|K| = k$ and $|M| = m$ such that $H = \{K \cup P \mid P \in \binom{M}{r}\}$. $K$ is called the kernel of $H$. Recently Ramsey problems for daisies were studied by Pudl{\'a}k, R{\"o}dl and Sales \cite{pudlak2022ramsey, sales2022ramsey}.  A balanced certificate for $f :\binom{[n]}{n/2} \rightarrow \{0,1\}$ of size $l$ can be viewed as a monochromatic $(\frac{n-l}{2}, n-l, \frac{l}{2})$-daisy in the coloring of $\binom{[n]}{n/2}$ corresponding to $f$. We note however that \cite{pudlak2022ramsey, sales2022ramsey} study the Ramsey problem for daisies with unrestricted kernel size or fixed kernel size $k$. On the other hand, we are interested in the case where $r$ and $m := 2r$ are given but $k = \frac{n}{2}-r$, as done in Subsection \ref{subs:lll}.

One simple question that we have been unable to resolve is whether for every $n$, the maximum query complexity (weakly) increases with the weight $k$, as $k$ varies from $1$ to $n/2$.

Another set of questions comes from trying to understand the relations and separations between complexity measures of functions on the balanced slice. Through their general framework, Dafni et al. \cite{dafni2021complexity} obtained polynomial relations between such measures. However some of the ideas used for proving relations between measures for the hypercube do not extend directly to their general setting. Looking specifically at the balanced slice can possibly improve those relations.

\vspace{10pt}

\noindent\textbf{Acknowledgements:} I wish to thank Yufal Filmus and Rajat Mittal for helpful discussions and encouragement. I am especially grateful to Yuval Filmus for suggesting minimum balanced certificate complexity. I also thank Manaswi Paraashar for several suggestions that helped improve the writing of this manuscript. Aaronson's Boolean Function Wizard $\cite{bfw, aaronson2003algorithms}$ helped in performing computer experiments.

\bibliographystyle{alpha}
\bibliography{refs}

\appendix

\section{$\Weights_{n,m,k}$} \label{sec:weights}
We now examine whether the strategies described in Section \ref{sec:comp} for the problem $\Weights_{n, m, k}$ of computing the multiset of weights of $x_i$'s are best possible. We assume that the total weight $k \leq N/2 = nm/2$ as before. If $k=1$ or $k = 0$, the function is constant. So we will have $k \geq 2$ from here on. We drop the subscript $k$ in the following proposition since it does not appear in the bound.

\begin{prop}
For $2 \leq k \leq nm/2$, $D(\Weights_{n, m}) \geq \min\{nm-3, (n-1)m - 1, n(m-1)\} = \min\{(n-1)m - 1, n(m-1)\}$.
\end{prop}
\begin{proof}
We will show that any query algorithm correctly computing $\Weights_{n, m}$ making less than $nm-3$ queries must make at least $\min\{(n-1)m - 1, n(m-1)\}$ queries. The adversary answers each query arbitrarily while respecting the weight condition and ensuring that at least $2$ ones and $2$ zeroes are left after each query. This can be done since the algorithm makes at most $nm-4$ queries by assumption and the weight $k$ satisfies $2 \leq k \leq N/2$.

Once all queries have been made, we may assume that for all $i < j$, the number of revealed bits of $x_i$ is at least that of $x_j$ by permuting the blocks if required. We may also assume that in each block $x_i$, an initial segment of bits has been queried and the later bits are all unqueried. It is convenient to imagine the revealed bits as lying inside a Ferrer's diagram which is enclosed inside an $m \times n$ box with the $i^{th}$ column representing $x_i$. 

Now we describe two cases in which the function value is not determined. 
\begin{itemize}
    \item Suppose the four bits in the $2 \times 2$ bottom right corner of the $m \times n$ box haven't been queried, i.e. the bits $x_{i, j}$ for $n-1 \leq i \leq n, m-1 \leq j \leq m$ are unqueried. Arbitrarily fix all remaining unqueried bits other than these $4$ while ensuring that $2$ ones and $2$ zeroes are left. Let $w_1$ and $w_2$ be the weights of $x_{n-1}$ and $x_{n}$ ignoring the unrevealed bits. Without loss of generality, suppose $w_1 \leq w_2$. When $x_{n-1, m-1} = x_{n, m-1} = 1$ and $x_{n-1, m} = x_{n, m} = 0$, the weights of $x_{n-1}$ and $x_n$ are $w_1+1$ and $w_2 + 1$ respectively. When $x_{n-1, m-1} = x_{n-1, m} = 0$ and $x_{n, m-1} = x_{n, m} = 1$, the weights of $x_{n-1}$ and $x_n$ are $w_1$ and $w_2 + 2$ respectively. Since $w_1 \leq w_2$, the multisets $\{w_1+1, w_2+1\}$ and $\{w_1, w_2 + 2\}$ are not equal. So if the algorithm computes $F$ correctly, it must have queried $x_{n-1, m-1}$ (by the way we have rearranged the blocks and the bits within each block).
    \item In the second case, suppose the bits $x_{n-2, m}, x_{n-1, m}, x_{n, m-1}, x_{n, m}$ have not been queried. Again arbitrarily fix the other unqueried bits while ensuring that $2$ ones and $2$ zeroes are left. Let $w_1$ and $w_2$ be the weights of $x_{n-1}$ and $x_{n}$ ignoring the unrevealed bits. Suppose $w_1 > w_2$. The two inputs with $x_{n-2, m} = 1, x_{n-1, m} = 0, x_{n, m-1} = 0, x_{n, m} = 1$ and $x_{n-2, m} = 1, x_{n-1, m} = 1, x_{n, m-1} = 0, x_{n, m} = 0$ do not have the same multiset of weights. Similarly for $w_1 \leq w_2$, the inputs with $x_{n-2, m} = 0, x_{n-1, m} = 1, x_{n, m-1} = 1, x_{n, m} = 0$ and $x_{n-2, m} = 0, x_{n-1, m} = 0, x_{n, m-1} = 1, x_{n, m} = 1$ do not have the same multiset of weights. So if the algorithm computes $F$ correctly, it must have queried at least one of these $4$ bits. By the rearrangement above, at least one of $x_{n-2, m}$ and $x_{n, m-1}$ must have been queried.
\end{itemize}

Next we combine the conclusions about the query algorithm obtained from the above cases. If it queries $x_{n-1, m-1}$ (from the first case) and $x_{n-2, m}$ (from the second case), it must have made at least $(n-2)m + m-1 = (n-1)m - 1$ queries. On the other hand, if it queries $x_{n, m-1}$ (from the second case), it must have made at least $(m-1)n$ queries. Together we have that the algorithm must have made at least $\min\{(n-1)m - 1, n(m-1)\}$ queries.
\end{proof}

When $k = nm/2$, we can improve the above lower bound to $\min\{nm-(m+1), nm - \ceil{\frac{n}{m}}\}$ which is almost tight. The dependence on the weight is necessary to some extent as will be shown later for $m = 2$.
\begin{prop}
For $k = nm/2$ and $n \geq 4$, $D(\Weights_{n, m, k}) \geq \min\{nm-(m+1), nm - \ceil{\frac{n}{m}}\}$. In particular, when $\ceil{\frac{n}{m}} > m$, $D(\Weights_{n, m, k}) = nm - \ceil{\frac{n}{m}}$.
\end{prop}
\begin{proof}
The adversary picks an arbitrary assignment $a: [n] \rightarrow \{0,1, \dots, m-1\}$ such that each number in $\{0,1, \dots, m-1\}$ is assigned to roughly the same number of $i$'s ($|a^{-1}(j)| \in \{\floor{\frac{n}{m}}, \floor{\frac{n}{m}} + 1\}$) and $\sum_{j = 0}^{m-1} j |a^{-1}(j)| = \frac{nm}{2} - \floor{\frac{n}{2}}$. 
For each $x_i$, the adversary will reveal the bits of $x_i$ in a way that the first $m-1$ bits revealed in $x_i$ have weight $a(i)$. The adversary also keeps track of a multiset $S$ of final bits. Initially $|S| = n$ and $S$ has $\floor{\frac{n}{2}}$ ones. A bit is revealed from $S$ whenever the algorithm queries the last bit of some $x_i$. It is always ensured that as long as $|S| > 4$, after using a bit from $S$, $S$ still contains at least two $1$s and two $0$s. When $|S| = 4$ and a final bit is queried, the adversary abandons the current strategy and starts using the above basic strategy of just ensuring that out of all unrevealed bits, at least two are $1$s and two are $0$s.

Now we compute the minimum number of queries made to correctly compute $F$ against this strategy. First consider the case where $|S| \geq 4$ after all the queries have been made. In this case, since we still have at least two $1$s and two $0$s, by the argument for the above basic strategy, at least one of $x_{n-2, m}$ and $x_{n, m-1}$ must have been queried. (We assume that we have rearranged the blocks and the bits within each block as above.) Since we know that $|S| \geq 4$, $x_{n-2, m}$ could not have been queried. Therefore $x_{n, m-1}$ must have been queried. By design, this implies that for each $x_i$, the first $m-1$ bits have weight $a(i)$. Now note that if there exist $i \neq i'$ such that $a(i) \neq a(i')$, and $x_{i, m}, x_{i', m}$ both have not been queried, then $F$ is not yet determined since $\{a(i)+1, a(i')\} \neq \{a(i), a(i') + 1\}$. So there must be at least $m-1$ numbers $j$ in $\{0,1, \dots, m-1\}$ such that all $x_{i, m}$ for $a(i) = j$ have been queried. This implies that in addition to the $n(m-1)$ queries required to query $x_{n, m-1}$, at least $n - \ceil{\frac{n}{m}}$ queries must have been made. So in this case, totally $nm - \ceil{\frac{n}{m}}$ queries are required.

Next suppose $|S| < 4$ after all the queries have been made, i.e. $x_{n-3, m}$ was queried. Once this query is made the adversary uses the basic strategy above. From what was shown earlier, one of $x_{n-2, m}$ and $x_{n, m-1}$ must have been queried, along with $x_{n-1, m-1}$. If $x_{n-2, m}$ and $x_{n, m-1}$ are queried, then at least $nm - (m+1)$ queries must have been made. If $x_{n, m-1}$ was queried, we get that at least $n(m-1) + n-3 = nm - 3$ queries were made.
\end{proof}

For $n = 2$, $D(\Weights_{2, m, k}) = m = N/2$. The upper bound comes from querying $x_1$ completely (algorithm $A$). For the lower bound, consider the case where only $k-1$ $1$s (and possibly some $0$s) have been revealed in $x_1$, only $0$s have been revealed in $x_2$ and none of $x_1, x_2$ has been queried completely. In this case, it is possible that $|x_1| = k, |x_2| = 0$ or $|x_1| = k-1, |x_2| = 1$.

\begin{prop}
Suppose $m = 2, n \geq 4$.
\begin{enumerate}
    \item If $2 \leq k \leq \floor{\frac{n}{2}}$, $D(\Weights_{n, 2, k}) = n + k - 1 < \floor{\frac{3n}{2}}$.
    \item If $\floor{\frac{n}{2}} + 1 \leq k \leq n$, $D(\Weights_{n, 2, k}) = \floor{\frac{3n}{2}}$.
\end{enumerate}
\end{prop}
\begin{proof}
\begin{enumerate}
    \item The upper bound follows from a modification of algorithm $B$. Query $x_{i, 1}$ for all $i \in [n]$. If $k$ $1$s are seen or no $1$s are seen, the weights are $1^k, 0^{n-k}$. Otherwise query $x_{j, 2}$ for all $j$ such that $x_{j, 1} = 1$. This uses at most $k-1$ additional queries. So totally at most $n+k-1$ queries are made.
    
    For the lower bound, the adversary uses the following strategy:
    \begin{itemize}
        \item If the block containing the current query has no revealed bit and at least $k-1$ other blocks have not been queried at all, answer $0$.
        \item If the block containing the current query has no revealed bit and less than $k-1$ other blocks have not been queried at all, answer $1$.
        \item If the other bit in the block queried is $1$, answer $0$.
        \item If the other bit in the block queried is $0$ and there are $n-k$ other blocks with two $0$s, answer $1$.
        \item Otherwise, answer $0$.
    \end{itemize}
    Note that the above strategy generates a legal input ($(00)^{n-k}, 01, (10)^{k-1}$) if all $2n$ bits are queried (after permuting the blocks and the bits within them if required). Suppose an algorithm computing $F$ makes $t$ queries on the above strategy. After all queries have been made and revealing a bit according to the above strategy in each block which has not been queried at all, it cannot be that there exist a block whose only revealed bit is $0$ and another block whose only revealed bit is $1$, since both $1^k, 0^{n-k}$ and $2, 1^{k-2}, 0^{n-k+1}$ would be possible answers. So either all blocks whose first bit was revealed to be $0$ are completely queried or all the other blocks are completely queried. In the latter case, since we first reveal at least $n-k+1$ $0$s before any $1$s, we must have $t \geq n-k+1 + 2(k-1) = n+k-1$.
    
In the other case, we claim that at least $\min\{2(n-k+1) + k-1 = 2n - (k-1), 2(n-1)\}$ queries must have been made. 
We will show that if a $1$ and two $0$s have not been revealed, then the function value is not determined. Consider the input $((00)^{n-k}, 01, {10}^{k-1})$. (The blocks and bits within them have been rearranged as done previously.) This has weights $0^{n-k}, 1^k$. On the other hand, we have the input $(00)^{n-k}, 01, (10)^{k-3}, 11, 00$. This input has weights $0^{n-k+1}, 1^{k-2}, 2$. To ensure, that it is not the case that a $1$ and two $0$s have not been revealed after all queries have been made, the algorithm must have queried at least one of $x_{n-1, 2}$ or $x_{n, 1}$. In the former case, at least $2(n-1)$ queries are made and in the latter, at least $2(n-k+1) + k-1$ queries are made. Now $2n - 2 \geq n+k-1$. Also $2n-(k-1) \geq n+k-1$ since $n \geq 2k$.
    
    \item The upper bound follows from algorithm $B$.
    
    For the lower bound, the adversary uses the above strategy with some small changes:
    \begin{itemize}
        \item If the block containing the current query has no revealed bit and at least $\floor{\frac{n}{2}}$ other blocks have not been queried at all, answer $0$.
        \item If the block containing the current query has no revealed bit and less than $\floor{\frac{n}{2}}$ other blocks have not been queried at all, answer $1$.
        \item If the other bit in the block queried is $1$, answer $0$.
        \item If the other bit in the block queried is $0$ and less than $n-k$ blocks have both $0$s, answer $0$.
        \item If the other bit in the block queried is $0$ and at least $n-k$ blocks have both $0$s, answer $1$.
    \end{itemize}
    The arguments for bounding the number of queries are the same as above. Here we just state the minimum number of queries required in the different cases. If all blocks whose first bit is $1$ are completely queried, at least $n + \floor{\frac{n}{2}} = \floor{\frac{3n}{2}}$ queries must have been made. In the other case, again since $\floor{\frac{n}{2}} \geq 2$ at least one of $x_{n-1, 2}$ or $x_{n, 1}$ must have been queried. So at least $\min\{2(n-2), n + \ceil{\frac{n}{2}}\} \geq \floor{\frac{3n}{2}}$ queries are made in this case as well.\qedhere    
\end{enumerate}
\end{proof}
\end{document}